\documentclass[11pt]{article}

\usepackage{amsfonts}
\usepackage{latexsym}
\usepackage{epsfig}

\newtheorem{theorem}{Theorem}
\newtheorem{lemma}{Lemma}
\newtheorem{proposition}{Proposition}
\newtheorem{corollary}{Corollary}
\newtheorem{observation}{Observation}
\newtheorem{claim}{Claim}
\newtheorem{fact}{Fact}

\def \remark {\noindent {\bf Remark.} \hskip 5pt}

\begin{document}

\title { Approximation algorithms for TSP with neighborhoods \\
in the plane }
\author {
{\sl Adrian Dumitrescu } 
\thanks {This work was done while the author
was a visiting faculty member at Stony Brook University.}\\
University of Wisconsin--Milwaukee  \\
Milwaukee, WI 53201-0784 \\
{\tt ad@cs.uwm.edu }
\and
{\sl Joseph S. B. Mitchell }
\thanks{Partially supported by grants from HRL Laboratories, NASA Ames, 
the National Science Foundation (CCR-9732220), Northrop-Grumman Corporation, 
Sandia National Labs, and Sun Microsystems.}\\
Stony Brook University \\
Stony Brook, NY 11794-3600 \\
{\tt jsbm@ams.sunysb.edu}
}

\date{August 9, 2014}
\maketitle

\begin{abstract}
  In the Euclidean TSP with neighborhoods (TSPN), we are given a
  collection of $n$ regions ({\em neighborhoods}) and we seek a
  shortest tour that visits each region. As a generalization of the
  classical Euclidean TSP, TSPN is also NP-hard.  In this paper, we
  present new approximation results for the TSPN, including (1) a
  constant-factor approximation algorithm for the case of arbitrary
  connected neighborhoods having comparable diameters; and (2) a PTAS
  for the important special case of disjoint unit disk neighborhoods
  (or nearly disjoint, nearly-unit disks).  Our methods also yield
  improved approximation ratios for various special classes of
  neighborhoods, which have previously been studied.  Further, we give
  a linear-time $O(1)$-approximation algorithm for the case of
  neighborhoods that are (infinite) straight lines.
\end{abstract}

\section { Introduction }

A salesman wants to meet a set of $n$ potential buyers. Each buyer
specifies a connected region in the plane, his {\em neighborhood},
within which he is willing to meet the salesman. For example, the
neighborhoods may be disks centered at buyers' locations, and the
radius of each disk specifies the maximum distance a buyer is willing
to travel to the meeting place. The salesman wants to find a tour of
shortest length that visits all of buyers' neighborhoods and finally
returns to its initial departure point. A variant of the problem,
which we will address in this paper, is that in which no departure
point is specified, and only a tour of the neighborhoods is to be
found.  This problem, which is known as the {\em TSP with
  neighborhoods} (TSPN), is a generalization of the classic Euclidean
Traveling Salesman Problem (TSP), in which the regions, or
``neighborhoods,'' are single points, and consequently is
NP-hard~\cite{gj-cigtn-79,p-etspi-77}.

\paragraph{Related Work.}\quad 
The TSP has a long and rich history of research in combinatorial
optimization.  It has been studied extensively in many forms,
including geometric instances; see
\cite{b-fagts-92,jrr-tsp-95,llrs-tsp-85,m-gspno-00,r-fhlgt-92}.  The
problem is known to be NP-hard, even for points in the Euclidean
plane~\cite{gj-cigtn-79,p-etspi-77}.  It has recently been shown that
the geometric instances of the TSP, including the Euclidean TSP, have a
polynomial-time approximation scheme, as developed by
Arora~\cite{a-nltas-98} and Mitchell~\cite{m-gsaps-99-new}, and later 
improved by Rao and Smith~\cite{rs-98}.

Arkin and Hassin~\cite{ah-aagcs-94} were the first to study approximation
algorithms for the geometric TSPN.  They gave $O(1)$-approximation
algorithms for several special cases, including parallel segments of
equal length, translates of a convex region, translates of a connected
region, and more generally, for regions which have diameter segments
that are {\em parallel} to a common direction, and the ratio between
the longest and the shortest diameter is bounded by a constant.

For the general case of connected polygonal regions, Mata and
Mitchell~\cite{mm-aagtn-95} obtained an $O(\log {n})$-approximation
algorithm, based on ``guillotine rectangular subdivisions'', with time
bound $O(N^5)$, where $N$ is the total complexity of the $n$ regions.
Gudmundsson and Levcopoulos \cite{gl-faatn-99} have recently obtained a
faster method, which, for any fixed $\epsilon>0$, is guaranteed to
perform at least one of the following tasks\footnote{One does not
know in advance which one will be accomplished.}: (1) it outputs a tour
of length at most $O(\log {n}) $ times optimum in time $O(n \log {n}
+N)$; (2) it outputs a tour of length at most $(1+\epsilon)$ times
optimum in time $O(N^3)$.  So far, no polynomial-time approximation
algorithm is known for general connected regions.

Recently it was shown that TSPN is APX-hard and cannot be approximated
within a factor of 1.000374 unless P=NP~\cite{bgklos-tspnv-02,gl-hrtsp-00}.
In fact, the inapproximability factor for the vertex cover problem 
in graphs with degree bounded by 5, as stated in \cite{bk-99},
and on which the result in \cite{gl-hrtsp-00} is based, implies
a factor larger than that.

In the time since this paper first appeared, it has been shown by de
Berg et. al~\cite{bgklos-tspnv-02} that the TSPN has an
$O(1)$-approximation algorithm in the case that the regions are
the regions are connected, disjoint, convex and 
{\em fat}.
Also, Schwartz and Safra~\cite{ss-catsp-02} have improved the lower
bounds of \cite{bgklos-tspnv-02,gl-hrtsp-00} on the hardness of
approximation of several variants of the TSPN problem, and 
Jonsson~\cite{j-tsplp-02} has given an $O(n)$ time 
$\sqrt{2}$-approximation algorithm for
the case of regions that are lines in the plane; our previous 
$O(n)$ time $\pi/2$-approximation algorithm appears in Section~\ref{sec:lines}.

\paragraph{Summary of Our Results.}\quad
In this paper, we obtain several approximation results on the 
geometric TSPN, including:
\begin{description}
\item[(1)]
We extend the approaches initiated in \cite{ah-aagcs-94} and
obtain the first $O(1)$-approximation algorithm for the TSPN having
connected regions of the same or similar diameter. This solves among
others, the open problem posed in \cite{ah-aagcs-94}, to provide a
constant-factor approximation algorithm for TSPN on segments of the
same length and {\em arbitrary} orientation.  
\item[(2)] We give a polynomial-time approximation scheme (PTAS)
  for the case of disjoint unit disks or the case of 
  nearly disjoint disks of nearly the same size.  The algorithm is
  based on applying the $m$-guillotine method with a new area-based
  charging scheme.  The fact that there is a PTAS for the case in
  which the neighborhoods are ``nice,'' with no point lying
in more than a constant number of neighborhoods, 
  should be contrasted with the fact that the TSPN on arbitrary
  regions is APX-hard.  The construction in the proof of
  \cite{gl-hrtsp-00} utilizes ``skinny'' neighborhoods, which
  intersect each other extensively.
\item[(3)] We also give modest improvements on earlier approximation
bounds in \cite{ah-aagcs-94} for the cases of parallel segments of equal
length, translates of a convex region, and translates of a connected
region.
\item[(4)] We present simple algorithms which achieve a
constant-factor guarantee for the case of equal disks and for the case
of infinite straight lines.
\end{description}

\paragraph{Preliminaries.}\quad
The input to our algorithms will be a set ${\cal R}$ of $n$ {\em
regions}, each of which is a simply-connected, closed subset of the
plane, $\Re^2$, bounded by a finite union of arcs of constant-degree
algebraic curves (degenerate regions are simply points).  
Since the regions are assumed to be closed, they
include the points that lie on the curves that form their boundary.
The assumption that regions are {\em simply} connected means that each
region has no ``holes.''  

Ideally, each region is a subset of the plane lying inside a simple,
closed, continuous curve, together with the curve itself. However, we
can only deal with regions that are computer-representable, hence the
above definition. 

Examples of allowable regions include simple
polygons, whose boundaries are unions of a finite number of straight
line segments, circular disks, regions bounded by straight segments
and circular arcs, infinite straight lines, etc.  We let $N$ denote
the total number of arcs specifying all $n$ regions in ${\cal R}$,
i.e. the total combinatorial complexity of the input.

A {\em tour} (or {\em circuit}) $T$, is a closed continuous curve
that visits each region of ${\cal R}$.
The {\em length} of tour $T$, denoted $|T|$, is the Euclidean
length of the curve $T$.  
To avoid ambiguity, the size of a finite set $X$ is denote by $\#X$.
In the {\em TSP with neighborhoods} (TSPN)
problem, our goal is to compute a tour whose length is guaranteed to
be close to the shortest possible length of a tour.  We let $T^*$
denote any optimal tour and let $L^*=|T^*|$ denote its length.  An
algorithm that outputs a tour whose length is guaranteed to be at most
$c\cdot L^*$ is said to be a {\em $c$-approximation algorithm} and to
have an {\em approximation ratio} of $c$.  A family of
$(1+\epsilon)$-approximation algorithms, parameterized by
$\epsilon>0$, and each running in polynomial time for fixed
$\epsilon$, is said to be a {\em polynomial-time approximation
scheme} (PTAS). 

\paragraph{Outline of the Paper.}\quad
In Section~\ref{sec:equal-disks} we use some simple packing arguments
to yield approximation algorithms for the TSPN for equal-size disks.
Section~\ref{sec:ptas} presents a PTAS for the TSPN for equal-size
disks.  In Section~\ref{sec:same-diameter} we give an approximation
algorithm for the TSPN for regions having the same diameter.  Finally,
in Section~\ref{sec:lines}, we give an approximation algorithm for the
case of regions that are infinite straight lines.  We conclude 
with a short list of open problems for future research.

\section {Equal-Size Disks}
\label{sec:equal-disks}

We begin by giving some simple arguments and corresponding algorithms
that achieve a constant approximation ratio for TSPN on a set ${\cal
  R}$ of $n$ disks of the same size.  Without loss of generality, we
assume that all disks have unit radius.  Our results carry over
naturally to disks of nearly the same size, with corresponding changes
in the approximation factor.

First, we consider the case of disjoint unit disks.  The algorithm is
simple and natural: using known PTAS results for TSP on
points~\cite{a-nltas-98,m-gsaps-99-new,rs-98}, compute, in time $O(n\log n)$, a
$(1+\epsilon)$-approximate tour, $T=T_C$, of the center points of the
$n$ disks.  We refer to it as the {\em center tour}. 
Clearly $T_C$ is a valid region tour.  We claim that

\begin{proposition}
Given set of $n$ disjoint unit disks, 
one can compute a tour, $T$, whose length satisfies
$$ |T| \leq ((1+\frac{8}{\pi})|T^*|+8)(1+\epsilon), $$
where $T^*$ is an optimal tour. The running time is dominated by
that of computing a $(1+\epsilon)$-approximate tour of $n$ points.
\end{proposition}

\begin{proof}
Put $L^*=|T^*|$.
Since $T^*$ visits all disks, the area $A_{L^*}$ swept by
a disk of radius 2, whose center moves along $T^*$, covers 
all of the unit disks ${\cal R}$. This area is bounded as follows
$$ \pi n \leq A_{L^*} \leq 4L^*+4\pi.$$
Thus, $n \leq 4 +\frac{4L^*}{\pi}$. 
A center tour of length at most 
$$ L^*+2n \leq L^*+2(4 +\frac{4L^*}{\pi})= (1+\frac{8}{\pi})L^*+8 $$ 
can be obtained by going along $T^*$ and making a detour of length at
most 2 to visit the center of a disk at each point where $T^*$ first visits a
disk. Hence, the length of the computed tour is bounded as claimed.
\end{proof}

For large $n$, the approximation ratio above is
$(1+\frac{8}{\pi})(1+\epsilon) \leq 3.55$; for small constant values
of $n$, the problem can be solved exactly using brute force.  We note
that for any algorithm that outputs a tour on the center points, we
cannot expect an approximation ratio smaller than $2$.  To see this,
consider a large square, and place almost touching unit disks along
its perimeter, both on its inside and on its outside. All of the disks
touch the perimeter, which is also an optimal disk tour except at the
four corners of the square. The length of the disk center tour is
roughly two times the perimeter of the square.

When the disks are nearly of the same size, so that the ratio
between the maximum and minimum radius is bounded by a constant $k>1$,
for large $n$, the approximation ratio is about $1+\frac{8 k^2}{\pi}$. 

Next we consider the case in which the disks can overlap.  First, we
compute a maximal independent (pairwise-disjoint) set $I$ of
disks. Next, we compute $T_I$, a $(1+\epsilon)$-approximate tour of
the center points of disks in $I$. Finally, we output a tour $T$
obtained by following the tour $T_I$, taking detours
around the boundaries of each of the
disks in $I$, as illustrated in
Figure~\ref{udisk}. More specifically, we select an arbitrary
disk $D_0$ and one of the intersection points, $s$, between $D_0$ and
$T_I$. We start at point $s$ and go clockwise along $T_I$.  Whenever the
boundary of a disk of $I$ is encountered, we follow clockwise around the boundary of the disk until we encounter again the tour $T_I$.  When we finally
reach $s$, we continue clockwise around $D_0$ to $T_I$ and then continue
counterclockwise around $T_I$, again taking detours clockwise around the
disks of $I$ that we encounter along the way.  
The tour $T$ finally ends when we return the second time to $s$.
In this way, our tour $T$ traverses the boundary of each disk of $I$ exactly
once and therefore visits all of the disks that are not in $I$ as
well. We remark that this method of constructing a feasible tour of the disks
results in a slightly better worst-case ratio than another
natural strategy for extending $T_I$ to a full disk tour: While traversing
$T_I$, each time one encounters the boundary of a disk $D \in I$, traverse the
entire circumference of $D$ exactly once, and then go directly
to the point where the tour $T_I$ exits $D$ and continue along $T_I$.

\begin{figure}[htbp]
\centerline {\input{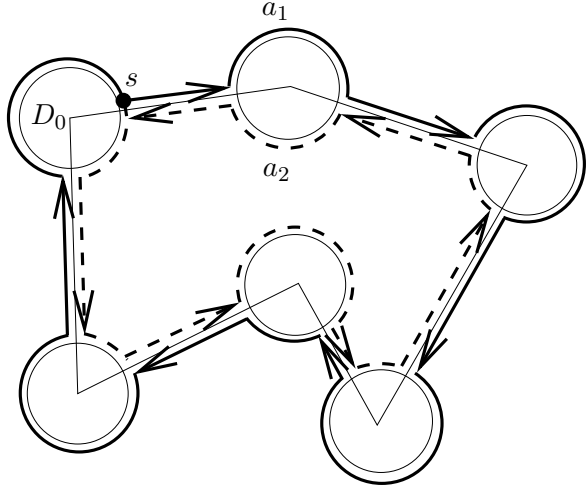}}
\caption{Construction of the tour $T$ from the
  tour $T_I$ on the center points of a maximal independent set, $I$,
  of disks.  Starting from $s$, we follow the thick solid tour, back
  to $s$, and then follow the thick dashed tour back around to $s$
  again.}
\label{udisk}
\end{figure}

Denote by $T^*$ an optimal disk tour, of length $L^*$, and by $T^*_I$
an optimal disk tour of the subset $I\subseteq {\cal R}$ of disks.  A
constant approximation ratio can be derived from the following three
inequalities:
$$ |T| \leq \pi|T_I| +2\pi ,   \eqno(1) \label{E1} $$
$$ |T_I| \leq   \left((1+\frac{8}{\pi}\right) |T^*_I|+8) (1+\epsilon) ,  
\eqno(2) \label{E2} $$
$$ |T^*_I| \leq L^* .   \eqno(3) \label{E3} $$
The third inequality follows from the fact that $I \subseteq {\cal R}$,
and the second from the case of disjoint unit disks considered above.
To check the first inequality, decompose $T_I$ into $\#I$ parts, 
assuming $\#I \geq 2$, one for each disk in $I$, by cutting each segment 
between two consecutive centers in the middle. Let $x+y=d \geq 2$ be the 
length of one of these parts of tour $T_I$, corresponding to a disk $D$, 
where $x,y$ are the lengths of the two segments of $T_I$ adjacent to 
the center of $D$. 
Write $a_1,a_2$ for the arc lengths of $D$ when its boundary is traversed
by $T$; we have $a_1+a_2=2\pi$. Writing the ratio
of the length of the corresponding part of $T$ to the length, $d$, of
this part of $T_I$, we get
$$ \frac{(x-1+a_1+y-1)+(x-1+a_2+y-1)}{d}=
\frac{2(d+\pi-2)}{d} \leq \pi, $$
the maximum being attained when $d=2$. When $\#I=1$, we have 
$|T_I|=0, |T|=2\pi$. Thus (1) is satisfied in all cases.
Putting (1), (2), (3) together, we get

\begin{proposition}
Given set of $n$ unit disks, possibly overlapping, 
one can compute a tour, $T$, whose length satisfies
$$ |T| \leq ((\pi+8)|T^*|+8\pi)(1+\epsilon)+2\pi, $$
where $T^*$ is an optimal tour. The running time is dominated by
that of computing a $(1+\epsilon)$-approximate tour of $n$ points.
\end{proposition}

For large $n$, the approximation ratio is $(\pi+8)(1+\epsilon) \leq 11.15$.
We note that the approximation ratio we have obtained with this
approach for disjoint unit disks, $3.55$, (resp., $11.15$ for unit
disks) is better (resp., weaker) than $\sqrt{3^2+7^2} \approx 7.62$,
the approximation ratio that will be given at the end of
Section~\ref{sec:same-diameter} for translates of a convex region,
which applies, of course, to the case of unit disks.

\section{A PTAS for Disjoint Equal Disks}
\label{sec:ptas}

In this section we present a polynomial-time approximation scheme
for the TSPN in the case of regions that are equal-size disks
or nearly equal-size disks.

Given the powerful methods that have been developed to obtain PTAS's
for various geometric optimization problems, such as the Euclidean
TSP, it is natural to suspect that these same techniques may apply to
the TSPN.  Indeed, one may expect that TSPN should have a PTAS based
on applying existing methods.  However, we know now, from the recent
APX-hardness result of \cite{gl-hrtsp-00}, that this cannot be.
What goes wrong?

The basic issue we must address in order to apply these techniques is
to be able to write a recursion to solve an appropriate ``succinct''
subproblem with dynamic programming.  What is the subproblem
``responsible'' for solving?  For problems involving {\em points}, the
subproblem can be made responsible for constructing some kind of
inexpensive network on the points {\em inside} the subproblem defined
by a rectangle, and to interconnect this network with the boundary in
some nicely controlled way, e.g., with only a constant complexity of
connection, in the case of $m$-guillotine methods.  The problem with
{\em regions} is that they can cross subproblem boundaries.  Then, we
do not know if the subproblem is responsible to visit the region, or
if the region is visited {\em outside} the subproblem.  We cannot
afford to enumerate the subset of regions that cross the boundary for
which the subproblem is responsible -- there are too many such
subsets, leading to too many subproblems.  Thus, we need a new idea.

Our approach is to employ a new type of structural result, based on
the general method of $m$-guillotine subdivisions.  In particular, we
show how to transform an optimal tour into one of a special class of
tours that recursively has a special $m$-guillotine structure,
permitting us to have a succinct, constant-size specification of the
subset of regions, crossing the subproblem boundary, for which the
subproblem is ``responsible'' in that it must visit these regions on
its interior.  In order to bound the increase in tour length in
performing this transformation, we must ``charge off'' the added tour
length to some small fraction of the length of the optimal tour, just
as is done in proving the bounds for the $m$-guillotine PTAS method
for TSP.  In order to do this charging, we must assume some special
structure on the class of neighborhoods in the TSPN, e.g., that the
regions ${\cal D}=\{D_1,\ldots,D_n\}$ are pairwise-disjoint,
equal-size disks, or have a similar
structure allowing us to relate tour length to area.

Here, we show how the approach applies to disjoint disks having equal
radii, $\delta$; generalizations are readily made to the case of
nearly equal radii, with a constant upper bound on the ratio of radii, 
and to the case of ``modestly overlapping,'' in which the disks become
pairwise-disjoint if all of them are decreased in size by a constant
factor while keeping the center points the same.

We begin with some definitions, largely following the notation of
\cite{m-gsaps-99-new}.  Let $G$ be an embedding of a planar graph, and
let $L$ denote the total Euclidean length of its edges, $E$. We can
assume without loss of generality that $G$ is restricted to the unit
square, $B$; i.e., $E\subset int(B)$.

Consider an axis-aligned rectangle $W$, a {\em window}, with
$W\subseteq B$.  Rectangle $W$ will correspond to a subproblem in a dynamic
programming algorithm.  Let $\ell$ be an axis-parallel line
intersecting $W$. We refer to $\ell$ as a {\em cut} and assume,
without loss of generality, that $\ell$ is vertical.

The intersection, $\ell\cap (E\cap int(W))$, of a cut $\ell$ with
$E\cap int(W)$ consists of a possibly empty set of subsegments of
$\ell$.  These subsegments are possibly singleton points.  Let $\xi$
be the number of endpoints of such subsegments along $\ell$, and let
the points be denoted by $p_1,\ldots,p_{\xi}$, in order of decreasing
$y$-coordinate along~$\ell$.  For a positive integer $m$, we define
the {\em $m$-span}, $\sigma_m(\ell)$, of $\ell$ with respect to $W$ as
follows.  If $\xi\leq 2(m-1)$, then $\sigma_m(\ell)=\emptyset$;
otherwise, $\sigma_m(\ell)$ is defined to be the line segment,
$p_{m}p_{\xi-m+1}$, joining the $m$th endpoint, $p_m$, with the
$m$th-from-the-last endpoints, $p_{\xi-m+1}$.  Refer to
Figure~\ref{fig:span}.  Note that the segment $p_{m}p_{\xi-m+1}$ may
be of zero length in case $p_{m}=p_{\xi-m+1}$.

\begin{figure}[htbp]
\centerline {\input{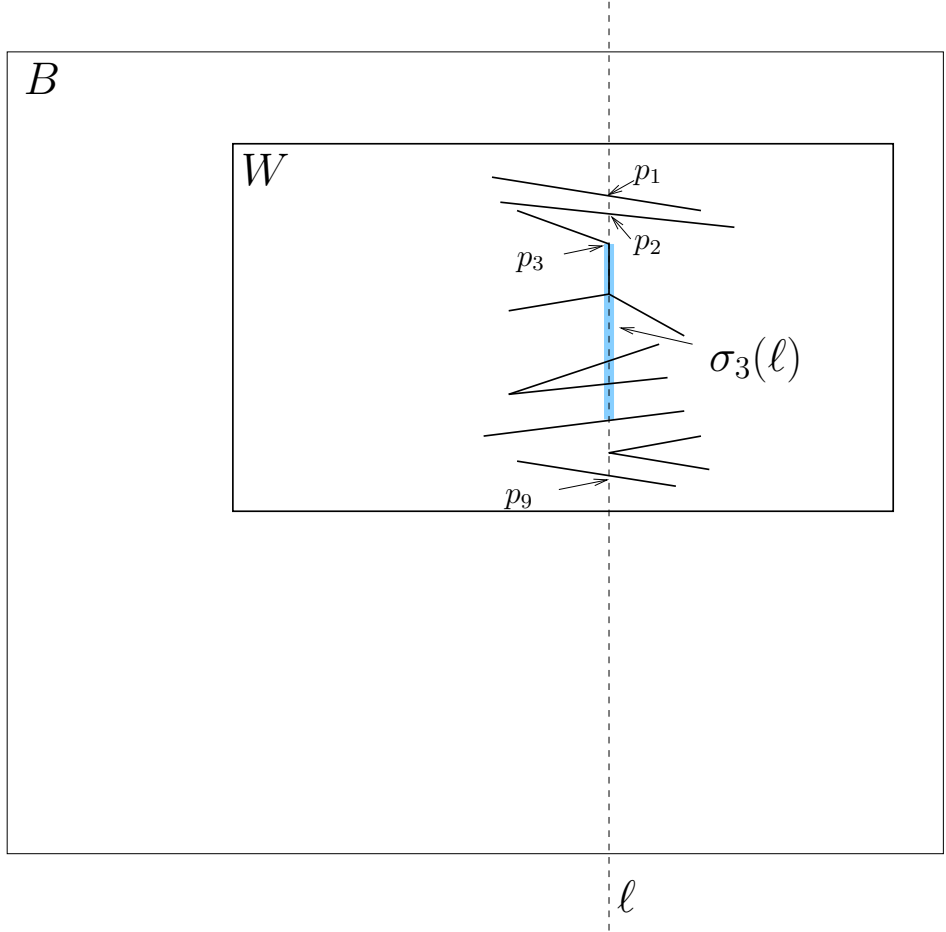}}
\caption{
  Definition of $m$-span: Here, the 3-span, $\sigma_3(\ell)$ of $\ell$
  with respect to the window $W\subset B$ is highlighted with a thick
  shaded vertical segment.}
\label{fig:span}
\end{figure}

\begin{figure}[htbp]
\centerline {\input{span-disk.pstex_t}}
\caption{
  Definition of $m$-disk-span: Here, the 3-disk-span, $\sigma_{3,{\cal
      D}}(\ell)$ of $\ell$ with respect to the window $W\subset B$ is
  highlighted with a thick shaded vertical segment.}
\label{fig:span-disk}
\end{figure}

The intersection, $\ell\cap {\cal D}\cap W$ of $\ell$ with the disks
that intersect $W$ consists of a possibly empty set of $\xi_{\cal
  D}\leq |{\cal D}\cap W|$ subsegments of $\ell$, one for each disk
that is intersected by $\ell$ in $W$.  Let these disks be
$D_1,D_2,\ldots,D_{\xi_{\cal D}}$, in order of decreasing
$y$-coordinate.  For a positive integer $m$, we define the {\em
  $m$-disk-span}, $\sigma_{m,{\cal D}}(\ell)$, of $\ell$ with respect
to $W$ as follows.  If $\xi_{\cal D}\leq 2m$, then $\sigma_{m,{\cal
    D}}(\ell)=\emptyset$; otherwise, $\sigma_{m,{\cal D}}(\ell)$ is
defined to be the possibly zero-length line segment joining the bottom
endpoint of $\ell\cap D_{m}$ with the top endpoint of $\ell\cap
D_{\xi_{\cal D}-m+1}$.  Refer to Figure~\ref{fig:span-disk}.

Line $\ell$ is an {\em $m$-good cut with respect to $W$} if
$\sigma_m(\ell)\subseteq E$ and $\sigma_{m,{\cal D}}(\ell)\subseteq
E$.  In particular, if $\xi\leq 2(m-1)$ and $\xi_{\cal D}\leq 2m$,
then $\ell$ is trivially an $m$-good cut, since both the $m$-span
and the $m$-disk-span are empty in this case.

We now say that $E$ satisfies the {\em $m$-guillotine property with
respect to window $W$} if either (1) $W$ does not fully contain any
disk; or (2) there exists an $m$-good cut, $\ell$, that splits $W$
into $W_1$ and $W_2$, and, recursively, $E$ satisfies the
$m$-guillotine property with respect to both $W_1$ and~$W_2$.

\begin{theorem}
\label{thm:main-guil}
Let $G$ be an embedded connected planar graph, with edge set $E$, of
total length $L$, and let ${\cal D}$ be a given set of
pairwise-disjoint equal-radius disks (of radius $\delta$) each of
which intersects $E$.  Assume that $E$ and ${\cal D}$ are contained in
the unit square $B$.  Then, for any positive integer $m$, there exists
a planar graph $G'$ that satisfies the $m$-guillotine property with
respect to $B$ and has an edge set $E'\supseteq E$ of length
$$L'\leq \left(1+{\sqrt{2}+16/\pi \over m}\right)L + {16\delta\over m}.$$
\end{theorem}

\begin{proof} 
  We convert $G$ into a new graph $G'$ by adding to $E$ a new set of
  horizontal/vertical edges whose total length is at most
  ${\sqrt{2}+16/\pi \over m}L + {16\delta\over m}$.  The construction
  is recursive: at each stage, we show that there exists a cut,
  $\ell$, with respect to the current window $W$ (which initially is
  the unit square $B$), such that we can afford to add both the
  $m$-span and the $m$-disk-span to~$E$.
  
  We say that a point $p$ on a cut $\ell$ is {\em $m$-dark} {\em with
    respect to $\ell$ and $W$} if, along $\ell^{\perp}\cap int(W)$,
  there are at least $m$ edges of $E$ intersected by $\ell^{\perp}$ on
  each side of $p$, where $\ell^{\perp}$ is the line through $p$ and
  perpendicular to $\ell$.\footnote{We can think of the edges $E$ as
    being ``walls'' that are not very effective at blocking light ---
    light can go through $m-1$ walls, but is stopped when it hits the
    $m$th wall; then, $p$ on a line $\ell$ is $m$-dark if $p$ is not
    illuminated when light is shone in from the boundary of $W$, along
    the direction of~$\ell^{\perp}$.}  We say that a subsegment of
  $\ell$ is {\em $m$-dark} (with respect to $W$) if all points of the
  segment are $m$-dark with respect to $\ell$ and~$W$.
  
  The important property of $m$-dark points along $\ell$ is the
  following: Assume, without loss of generality, that $\ell$ is
  horizontal.  We consider any line segment that lies along an edge of
  $E$ to have a {\em top} side and a {\em bottom} side; the top is the
  side that can be seen from above, from a point with $y=+\infty$.
  Then, if all points on subsegment $pq$ of $\ell$ are $m$-dark, we
  can charge the length of $pq$ off to the bottoms of the first $m$
  subsegments, $E^+\subseteq E$, of edges that lie above $pq$, and the
  tops of the first $m$ subsegments, $E^-\subseteq E$, of edges that
  lie below $pq$, since we know that there are at least $m$ edges
  ``blocking'' $pq$ from the top/bottom of $W$.  We charge $pq$'s
  length half to $E^+$, charging each of the $m$ levels of $E^+$ {\em
    from below}, with ${1\over 2m}$ units of charge, and half to
  $E^-$, charging each of the $m$ levels of $E^-$ {\em from above},
  with ${1\over 2m}$ units of charge.  We refer to this type of charge
  as the ``red'' charge.
  
  We say that a point $p$ on a cut $\ell$ is {\em $m$-disk-dark} {\em
    with respect to $\ell$ and $W$} if, along $\ell^{\perp}\cap
  int(W)$, there are at least $m$ disks of ${\cal D}$ that have a
  nonempty intersection with $\ell^{\perp}\cap W$ on each side of $p$.
  Here, again, $\ell^{\perp}$ is the line through $p$ and
  perpendicular to $\ell$.  We say that a subsegment of $\ell$ is {\em
    $m$-disk-dark} with respect to $W$ if all points of the segment
  are $m$-disk-dark with respect to $\ell$ and~$W$.  The {\em
    chargeable} length within $W$ of a cut $\ell$ is defined to be the
  sum of the lengths of its $m$-dark portion and its $m$-disk-dark
  portion.  Refer to Figure~\ref{fig:dark-disk}.
  
\begin{figure}[htbp]
\centerline {\input{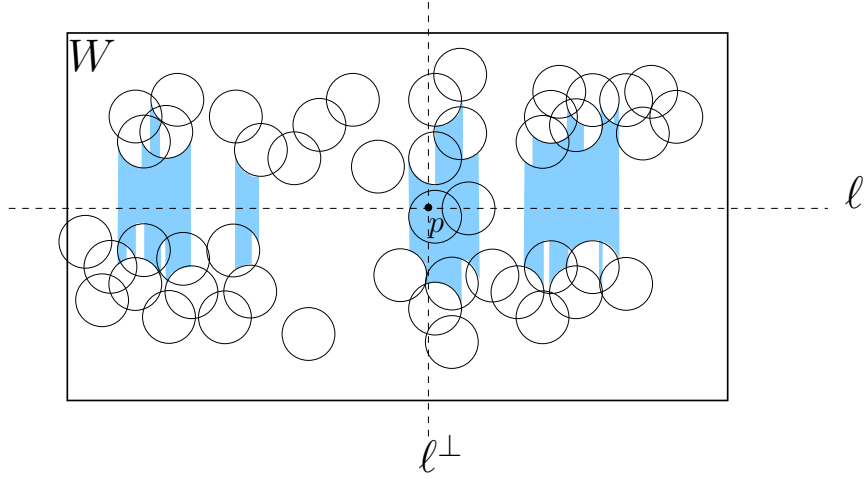}}
\caption{
  Definition of $m$-disk-dark: Here, the points that are 2-disk-dark
  with respect to $\ell$ are those four subsegments of $\ell$ that lie
  within the shaded regions, which comprise the set $R_x^{(2,{\cal
      D})}$ of points of $W$ that are 2-disk-dark with respect to
  horizontal cuts.}
\label{fig:dark-disk}
\end{figure}

  The important property of $m$-disk-dark points along horizontal
  $\ell$ is the following: If all points on subsegment $pq$ of $\ell$
  are $m$-disk-dark, then we can charge the length of $pq$ off to the
  bottoms of the first $m$ disks that lie above $pq$, and the tops of
  the first $m$ disks that lie below $pq$, since we know that there
  are at least $m$ disks ``blocking'' $pq$ from the top/bottom of
  $W$.  We charge $pq$'s length half upwards, charging the bottoms of
  each of the $m$ ``levels'' of disks with ${1\over 2m}$ units of
  charge, and half downwards, charging the tops of each of the $m$
  ``levels'' of disks with ${1\over 2m}$ units of charge.  We refer
  to this type of charge as the ``blue'' charge.
  
  A cut $\ell$ is {\em favorable} if its chargeable length within $W$
  is {\em at least as large as} the sum of the lengths of the $m$-span
  and the $m$-disk-span.  We do not attempt to take advantage of the
  fact that there may be overlap among the $m$-span and the
  $m$-disk-span, nor that there may be portions of these two segments
  that are already part of $E$; an argument taking advantage of these
  facts may improve slightly the constants in some of our bounds.
  
  The existence of a favorable cut is guaranteed by the following key
  lemma, whose proof is similar to that of the key lemma
  in~\cite{m-gsaps-99-new}:

\begin{lemma}
\label{lem:key}
For any $G$ and any window $W$, there is a favorable cut.
\end{lemma}

\begin{proof}
We show that there must be a favorable cut that
is either horizontal or vertical.

Let $f(x)$ denote the ``cost'' of
the vertical line, $\ell_x$, passing through the point $(x,0)$,
where ``cost'' means the sum of the lengths of the
$m$-span and the $m$-disk-span for $\ell_x$.
Then,
$$f(x)=|\sigma_m(\ell_x)| + |\sigma_{m,{\cal D}}(\ell_x)|.$$
Thus, 
$$A_x=\int_0^1 f(x) dx = A^{(m)}_x+A^{(m,{\cal D})}_x= \int_0^1
|\sigma_m(\ell_x)| dx + \int_0^1 |\sigma_{m,{\cal D}}(\ell_x)| dx,$$ 
where $A^{(m)}_x=\int_0^1 |\sigma_m(\ell_x)| dx$ is the area of the
$x$-monotone region $R^{(m)}_x$ of points of $B$ that are $m$-dark
with respect to horizontal cuts, and $A^{(m,{\cal D})}_x=\int_0^1
|\sigma_{m,{\cal D}}(\ell_x)| dx$ is the area of the $x$-monotone
region $R^{(m,{\cal D})}_x$ of points of $B$ that are $m$-disk-dark
with respect to horizontal cuts.  
Refer to Figure~\ref{fig:dark-disk}. Similarly, define $g(y)$ to be the
cost of the horizontal line through $y$, and let $A_y=\int_0^1 g(y)
dy$.

Assume, without loss of generality, that $A_x\geq A_y$.  We claim that
there exists a horizontal favorable cut; i.e., we claim that there
exists a horizontal cut, $\ell$, such that its chargeable length is at
least as large as the cost of $\ell$, meaning that the length of its
$m$-dark portion plus its $m$-disk-dark portion is at least
$|\sigma_m(\ell)|+|\sigma_{m,{\cal D}}(\ell)|$.  To see this, note
that $A_x$ can be computed by switching the order of integration,
``slicing'' the regions $R^{(m)}_x$ and $R^{(m,{\cal D})}_x$
horizontally, rather than vertically; i.e., $A_x=\int_0^1 h(y)
dy=\int_0^1 h_m(y) dy+\int_0^1 h_{m,{\cal D}}(y) dy$, where $h_m(y)$
is the $m$-dark length of the horizontal line through $y$, $h_{m,{\cal
    D}}(y)$ is the length of the intersection of $R^{(m,{\cal D})}_x$
with a horizontal line through $y$, and $h(y)$ is the chargeable
length of the horizontal line through $y$.  In other words, $h_m(y)$
(resp., $h_{m,{\cal D}}(y)$) is the length of the $m$-dark (resp.,
$m$-disk-dark) portion of the horizontal line through $y$.  Thus,
since $A_x\geq A_y$, we get that $\int_0^1 h(y) dy \geq \int_0^1 g(y)
dy\geq 0$.  Thus, it cannot be that for all values of $y\in[0,1]$,
$h(y)<g(y)$, so there exists a $y=y^*$ for which $h(y^*)\geq g(y^*)$.
The horizontal line through this $y^*$ is a cut satisfying the claim
of the lemma.

If, instead, we had $A_x\leq A_y$, then we would get a
{\em vertical} cut satisfying the claim.
\end{proof}

Now that we know there must be a favorable cut, $\ell$, we can charge
off the cost of the $m$-span and the $m$-disk-span of $\ell$, making
``red'' charge on the bottoms (resp., tops) of segments of $E$ that
lie above (resp., below) $m$-dark points of $\ell$, and making
``blue'' charge on the bottoms (resp., tops) of disks that lie above
(resp., below) $m$-disk-dark points of~$\ell$.
We then recurse on each side of the cut, in the two new windows.

After a portion of $E$ has been charged {\em red} on one side, due to
a cut $\ell$, it will be within $m$ levels of the boundary of the
windows on either side of $\ell$, and, hence, within $m$ levels of the
boundary of any future windows, found deeper in the recursion, that
contain the portion.  Thus, no portion of $E$ will ever be charged
{\em red} more than once from each side, in each of the two
directions, horizontal or vertical, so no portion of $E$ will ever pay
more than $\sqrt{2}/m$ times its length in red charge.  We charge at
the rate of ${1\over 2m}$ per unit length of the perimeter of the
segment's axis-aligned bounding box, and the worst case is achieved
for a segment of slope $\pm 1$.  Thus, the total red charge is at most
${\sqrt{2}\over m}L$.

Similarly, no disk will ever have its boundary charged {\em blue} more
than once from each of the two directions, horizontal or vertical.
Since it is charged at the rate of ${1\over 2m}$ per unit length of
its axis-aligned bounding box, whose perimeter is $8\delta$, we get a
total blue charge of at most ${4\delta \over m}n$.  We now appeal to
the lower bound from the previous section, which was based on an area
argument.  Here, for radius $\delta$ disks, that argument shows that
$\pi\delta^2 n\leq A_{L} \leq 4\delta L+(2\delta)^2\pi$, from which we
get $n\leq {4\over \pi\delta} L+ 4$.  Note that this area argument uses the
fact that $E$ is connected.  Thus, the total blue charge is at most
${16\over \pi m}L + {16\delta\over m}$.

It is also important to note that we are always charging red portions
of the original edge set $E$: the new edges added are never
themselves charged, since they lie on window boundaries and cannot
therefore serve to make a portion of some future cut $m$-dark.

Overall, then, the total increase in length caused by adding
the $m$-spans and $m$-disk-spans along favorable cuts
is bounded by 
$${16+\pi\sqrt{2}\over \pi m}L + {16\delta\over m}.$$

Our goal in adding the $m$-disk-span is to obtain a succinct
representation of which disks that straddle the boundary of a window
are visited within the window and which are visited outside the
window.  The $m$-disk-span segment visits all but a constant, $O(m)$, 
number of the disks on the corresponding side of the window.  There
is, however, one remaining issue with respect to the $m$-disk-span
segments: We need to argue that we can ``afford,'' within our charging
scheme, to {\em connect} the $m$-disk-span to the input edge set, $E$.
This is because, in our dynamic programming optimization, we will
find a shortest possible 
planar graph with the $m$-guillotine property that obeys certain
connectivity constraints, as well as other properties that guarantee
that the graph has an Eulerian subgraph spanning all disks.  The
optimal graph that we compute uses the $m$-disk-span segments to visit
the corresponding disks on the boundaries of windows that define
subproblems.

\remark
{\em 
(Clarification added after journal publication.)
We had previously phrased the above as ``find a shortest possible connected planar graph with the
$m$-guillotine property''; we have rephrased to be``find a shortest possible 
planar graph with the $m$-guillotine property that obeys certain connectivity constraints''.
}

In particular, we add connections to $E$ from the endpoints of the
$m$-disk-span to the point of $E$ in the corresponding disk that is
closest to the endpoint.  We know that this connection is of length at
most $2\delta$ per endpoint of the $m$-disk-span, since the disks
are of diameter $2\delta$.  In total, this adds only $4\delta$ to the
length of the $m$-disk-span.  Assuming the $m$-disk-span stabs at
least three or more disks, its length is at
least $\Omega(\delta)$, implying that we can charge off this extra
$4\delta$ for the connections in the same way that we charge off the
$m$-disk-spans themselves.  If, on the other hand, the $m$-disk-span
stabs only one or two disks, then we can afford to skip the addition
of the $m$-disk-span altogether, and just keep track in the dynamic
program of the necessary information for these couple extra disks, 
specifying whether they are to be visited within the window or not.

\remark
{\em (Clarification added after journal publication; thanks to
  Sophie Spirkl for her inquiry and input.)  We argue above that we
  can afford to add connections to $E$, since the length added is
  proportional to the bridge length.  Our dynamic programming
  algorithm computes a minimum-length planar network with the
  $m$-guillotine property that obeys certain connectivity constraints,
  so that all of the network is connected except (possibly) the
  $m$-disk-spans, which may not be connected to the rest of the
  (connected) network.  The objective function in the dynamic program
  requires that we minimize the total length of the network, counting
  the lengths of the bridges that serve as $m$-disk-spans a constant
  number of times.  Then, since we know we can afford to add length
  proportional to the total $m$-disk-span lengths, we know we can add
  the connections mentioned above, resulting in an overall connected
  network (and can make it Eulerian with appropriate doubling of the
  bridge segments, in the usual way, as mentioned again below),
  appropriately close to the optimal length.  
}
\end{proof}

\begin{corollary}
\label{cor:ptas}
The TSPN for a set of disjoint equal-size disks has a PTAS.
The same is true for a set of nearly disjoint, nearly equal-size
disks, for which there is a constant upper bound on the ratio of 
largest to smallest radius and there is a constant factor such
that the disks become disjoint if their radii are each multiplied 
by the factor while keeping the center points the same.
\end{corollary}

\begin{proof}
  We consider only the case of disjoint disks each of radius $\delta$;
  the generalization to nearly disjoint, nearly equal-size disks is
  straightforward.

We impose a regular $m$-by-$m$ grid on each disk; let ${\cal G}$
denote the resulting set of $O(m^2n)$ grid points.

Consider an optimal tour, $OPT$, of length $L^*$.  Now, $OPT$ is a
simple polygon.  We can perturb the vertices of $OPT$ so that each
lies at a grid point in ${\cal G}$, resulting in a new tour, $OPT'$,
visiting every disk, whose length is at most $L^*+O(n\delta/m)$.
Using the fact, from the previous section, that $n\leq {4\over \pi\delta} L^*+ 4$, and assuming
that $n\geq 8$, we get
that the length of $OPT'$ is at most $L^*+O({4\over \pi
  m}L^*+{4\delta\over m})\leq L^*(1+O({8\over \pi m}))$.
If $n<8$, then we can solve the problem in constant time by brute force.

Theorem~\ref{thm:main-guil} implies that we
can convert $OPT'$, which consists of some set $E$ of edges, 
into a planar graph, $OPT''$, obeying the $m$-guillotine property,
while not increasing the total length by too much.
In particular, the length of $OPT''$ is 
at most 
$$(1+{\sqrt{2}+16/\pi \over m})L^*(1+O({8\over \pi m})) + {16\delta\over m}\leq (1+O(1/m))L^*.$$

We now apply a dynamic programming algorithm, running in $O(n^{O(m)})$
time, to compute a minimum-length planar graph having a prescribed set
of properties: (1) it satisfies the $m$-guillotine property, which is necessary
for the dynamic program to have the claimed efficiency; (2) it visits
at least one grid point of ${\cal G}$ in each region $D_i$; and (3) it
contains an Eulerian subgraph that spans the disks. This third condition
that allows us to extract a tour in the end.  We only outline here
the dynamic programming algorithm; the details are very similar to
those of \cite{m-gsaps-99-new}, with the modification to account for the
$m$-disk-span.

A subproblem is defined by a rectangle $W$ whose coordinates
are among those of the grid points ${\cal G}$, together with a constant
amount ($O(m)$) of information about how the solution to the subproblem
interacts across the boundary of $W$ with the solution outside of $W$.
This information includes the following:
\begin{description}
\item[(a)] For each of the four sides of $W$, we specify a ``bridge''
  segment and at most $2m$ other segments with endpoints among ${\cal
    G}$ that cross the side; this is done exactly as in the case of
  the Euclidean TSP on points, as in \cite{m-gsaps-99-new}.
\item[(b)] For each of the four sides of $W$, we specify a ``disk
  bridge'' segment corresponding to the $m$-disk-span, and, for each
  of at most $2m$ disks that are not intersected by the disk bridge
  segment, we specify in a single bit whether the disk is to be
  visited within the subproblem or not; if not, it is visited outside
  the window~$W$.
\item[(c)] We specify a required ``connection pattern'' within $W$.
  In particular, we indicate which subsets of the $O(m)$ specified
  edges crossing the boundary of $W$ are required to be connected
  within $W$.  This, again, is done exactly as is detailed for the
  Euclidean TSP on point sets in \cite{m-gsaps-99-new}.
\end{description}

In order to end up with a graph having an Eulerian subgraph spanning
the disks, we use the same trick as done in \cite{m-gsaps-99-new}: we
``double'' the bridge segments and the disk bridge segments, and then
require that the number of connections on each side of a bridge
segment satisfy a parity condition.  Exactly as in
\cite{m-gsaps-99-new}, this allows us to extract a tour from the
planar graph that results from the dynamic programming algorithm, 
which gives a shortest possible graph that obeys the specified
conditions.
  
The result is that in polynomial time ($O(n^{O(m)})$) one can compute
a shortest possible graph, from a special class of such graphs, and
this graph spans the regions ${\cal D}$.  Theorem~\ref{thm:main-guil}
guarantees that the length of the resulting graph is very close, 
within factor $1+O(1/m)$, to the length of an optimal solution to the
TSPN.  
(We also know from the remarks previously, that we can afford to
add connections to assure that the $m$-disk-spans are connected to the
rest of the (connected) network.)
Thus, once we extract a tour from the Eulerian subgraph, we
have the desired $(1+\epsilon)$-approximation solution, where
$\epsilon=O(1/m)$.
  
Finally, we mention that the running time of $O(n^{O(m)})$ can be
improved to $O(n^{C})$, for a constant $C$ independent of $m$, using
the notion of grid-rounded guillotine subdivisions, as
in~\cite{joe-tsp-new,joe-cccg}; here, the dependence on the constant
$m$ is exponential in the multiplicative constant concealed by the
big-Oh of $O(n^{C})$.  We suspect that the techniques of Rao and
Smith~\cite{rs-98}, based on Arora's method~\cite{a-nltas-98}, can be
used to improve the time bound to $O(n\log n)$, while possibly also
addressing the problem in higher dimensions; we leave this for future
work.
\end{proof}

\section {Connected Regions of the Same Diameter} 
\label{sec:same-diameter}

In this section we give a constant-factor approximation algorithm
for the TSPN problem that applies in the case that all regions
have the same diameter or nearly the same diameter.

The {\em diameter} of a region, $\delta$, is the distance between
two points in the region that are farthest apart. 
Without loss of generality, we assume that all regions have unit diameter, 
$\delta=1$.
The general method we use is to carefully select a representative
point in each region, and then compute an almost optimal tour
on these representative points. This approach was initiated in 
\cite{ah-aagcs-94}.
We also employ a specialized version of the Combination
Lemma in \cite{ah-aagcs-94}.
Now we describe the algorithm, which is similar in many aspects to the
one in \cite{ah-aagcs-94}.

In each region, compute a unit-length diameter segment; in case of
multiple such segments, select one arbitrarily.
Computing a diameter segment of a region can be done
efficiently, in time linear in the complexity of the region.
Classify the regions into two types: (1) those for which the selected
diameter is {\em almost horizontal}, by which we mean its slope is
between $-45^{\circ}$ and $45^{\circ}$; (2) those for which the
selected diameter is {\em almost vertical}, by which we mean all the
others.  
Use algorithm {\bf A} (below) for each of
these two region types.  We will prove that a constant ratio is
achievable for each class, after a suitable transformation is applied
to regions of type (2). We then apply the Combination Lemma to
obtain a constant-factor approximation for all regions.

\begin{lemma} {\rm (Combination Lemma)} \label {L1} 
Given regions that can be partitioned into two types ---
those having almost horizontal unit diameters,
and those having almost vertical unit diameters respectively ---
and constants
$c_1, c_2$ bounding the error ratios with which we can approximate
optimal tours on regions of types $1$ and $2$, then we can 
approximate the optimal tour on all regions with an error ratio 
bounded by $c_1+c_2+2$.
\end{lemma} 

We omit the proof, which is a simplified version of the argument in
\cite{ah-aagcs-94}. The bound on the approximation ratio is still the
same. 
We remark that the Combination Lemma in \cite{ah-aagcs-94} implicitly
assumes that for each of the two types, the diameters of the regions 
are parallel to some direction, an assumption which does not hold in our case. 

Algorithm {\bf A} gives a constant-factor approximation of the optimal
tour for regions of type (1). Regions of type (2) are readily handled
by rotating them by $90^{\circ}$ to obtain type (1) regions.
When the diameters are nearly the same, so that the
ratio of the largest to the smallest is bounded by a constant,  
we can still get an approximation algorithm with a constant ratio;
however since this ratio is rather large even for same diameters, we
omit the calculations.

A set of lines is a {\em cover} of a set of regions if each region is
intersected by at least one line from the covering set. 
We refer to such a set of lines as {\em covering lines}.\\

\begin{figure}[htbp]
\centerline{\epsfysize=3.5in \epsffile{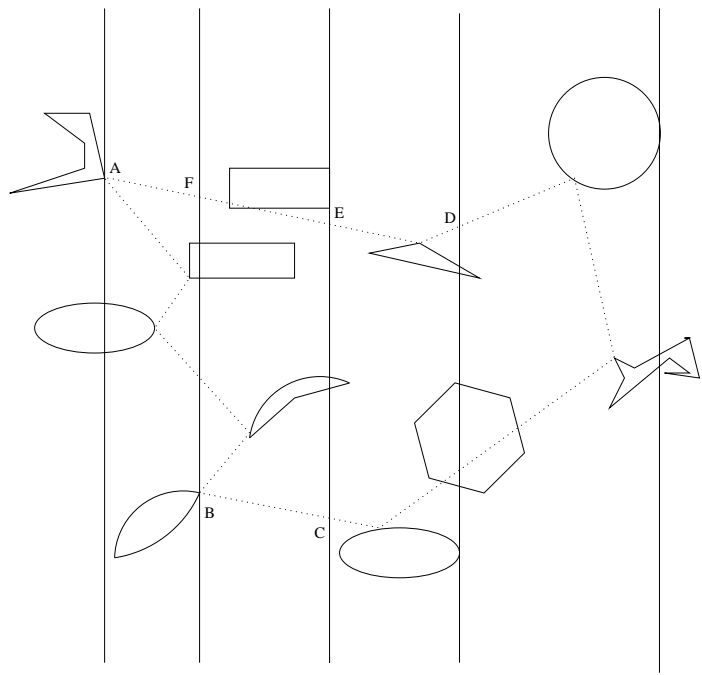}}
\caption{Illustration of Step 1 in Algorithm {\bf A}}
\label{cover}
\end{figure}

\noindent {\bf Algorithm A}. Input: a set of $n$ regions of type (1).\\

\noindent {\em Step 1}. Construct a cover of the regions by a 
minimum number of vertical lines. 
This procedure works in a greedy fashion,
namely the leftmost line is as far right as possible, so that
it is a right tangent of some region. To obtain this cover,
the intervals of projection on the $x$-axis of all regions are
computed and a greedy cover of this set of intervals is found.
After removing all intervals covered by a previous line, 
and there are still uncovered intervals,
another covering line is repeatedly added to the cover.
At the same time, a representative point for each region is 
arbitrarily selected on the corresponding covering line
and inside the region, e.g., the topmost boundary point of 
intersection between the region and its covering line. 
In this way $n$ representative points, one per region, are selected.
An illustration of this procedure appears in Figure \ref{cover}.
An important remark is that the representative points are not necessarily
on the selected diameters, since a diameter may not be entirely 
contained in its region.

Assuming that we have only regions of type (1), the greedy cover has
the effect of obtaining a large enough horizontal distance between any
two consecutive covering lines.  We remark that the greedy covering
algorithm for a set of closed intervals on a line is known to output a
cover of minimum size, a property which carries
over to  our vertical line cover.\\

\noindent {\em Step 2.} Proceed according to the following three cases.

\begin{description}
\item[Case 1:] The greedy cover contains one covering line.

Compute a smallest perimeter axis-aligned rectangle $Q$
that intersects all regions,
where $Q$ is considered a two-dimensional domain.  
Let $w$ and $h$ denote $Q$'s width and height respectively.
Consider $Q$ to be a graph with four vertices and four edges (its
four sides). Let $e_1$ and $e_2$ be the vertical segments
of height $h$ that partition $Q$ into three equal-width parts.
Add to the graph $Q$ a double edge corresponding to $e_1$ and
a doubled edge corresponding to $e_2$.  The resulting  
graph $G$ is an Eulerian multigraph, since all node
degrees are even, with $8$ vertices and $12$ edges. 
Output any Euler tour $T$ of $G$. While such a tour is not, in general,
the shortest possible tour visiting the regions, it suffices for
our purposes of approximation.

\item[Case 2:] The greedy cover contains two covering lines.

Move (if possible) the rightmost vertical covering line to the left 
as much as possible (while still covering all regions). Recompute
the representative points obtained in this way. Set $D$
to be the distance between the two covering lines (clearly $D>0$). 

\begin{description}
\item[Case 2.1:] $D \geq 3$. 

Compute an axis-aligned rectangle $Q$ of width $w=D$, with its vertical
sides along the two covering lines, and of minimal height, which
includes all representative points (on the two covering lines). 
Let $h$ denote $Q$'s height.
Output the tour $T$ that is the perimeter of $Q$.

\item[Case 2.2:] $D \leq 3$ (similar to Case 1).

Compute a smallest perimeter axis-aligned rectangle $Q$
that touches (intersects) all regions,
where $Q$ is considered to be a two-dimensional domain.  
Let $w$ and $h$ denote $Q$'s width and height respectively.
Note that $w\leq 5$.  
Let $e_1,\ldots,e_7$ be the vertical segments
of height $h$ that partition $Q$ into eight equal-width parts.
Consider the edges of $Q$, together with doubled
copies of the edges $e_1,\ldots,e_7$, to define an Eulerian multigraph, $G$,
having 18 vertices and 32 edges.
Output any Euler tour $T$ of $G$. 
\end{description}

\item[Case 3:] The greedy cover contains at least three covering lines.

Compute $T$, a $(1+\epsilon)$-approximate tour of the representative
points as the output tour.
\end{description}

If the regions are simple polygons, then a minimum-perimeter
touching rectangle $Q$ is determined by at most four contact points
with region boundary arcs.
A brute force procedure examining all possible
$k$-tuples, $k \leq 4$, of such arcs computes $Q$ in $O(N^5)$ time.  
The total running time of the approximation algorithm is either
bounded by the complexity of the above step or by the complexity of 
computing $(1+\epsilon)$-approximate tours on $n$ points, depending on
the size of the greedy cover.  

\begin{theorem}
\label{thm:main-approx}
Given a set ${\cal R}$ of $n$ connected regions of the same diameter
in the plane, a $O(1)$-approximation of an optimal tour can be
computed in polynomial time.
\end{theorem}

\begin{proof}
Let $OPT$ be an optimal region tour.  We address each of the cases we
distinguished in the previous algorithm. We will use repeatedly the
following simple fact (see \cite{ah-aagcs-94}):
For positive $a,b,w,h$ the following inequality holds
$$ aw+bh \leq \sqrt{a^2+b^2} \sqrt{w^2+h^2}. \eqno(4) $$

\noindent {\bf Case 1}. Write $diag(Q)$ for the diagonal of the 
rectangle $Q$. We first argue that $T$ visits all regions.
Since all regions are covered by a unique covering line,
they lie in a vertical strip of width $\leq 2\delta=2$. So $w \leq 2$.
The horizontal projection of each type (1) region (on the $x$-axis) is
at least $1/\sqrt{2}$. Hence, each region is intersected either
by the boundary of $Q$ or by one of the two vertical segments 
inside  $Q$, since these segments partition $Q$ into
three subrectangles each of width at most $2/3<1/\sqrt{2}$.
(Each region that is not intersected by the perimeter of $Q$ 
lies entirely inside $Q$). Consequently, $T$ is a valid region tour. 
We first give a lower bound on $|OPT|$.
Since the length of a tour touching all four sides of a rectangle is
at least twice the length of the diagonal of the rectangle (see
\cite{ah-aagcs-94}), 
$$ |OPT| \geq 2 diag(Q)=2\sqrt{w^2+h^2}. $$
The length of $T$ is 
$$ |T|=2w+6h=2(w+3h) \leq 2\sqrt{10} \sqrt{w^2+h^2} \leq \sqrt{10}|OPT|. $$

\noindent {\bf Case 2.1:} $D \geq 3$. We distinguish two sub-cases.

\qquad {\bf Case 2.1.a:} $h \leq 2$. Recall that $w=D \geq 3$.
$$ |OPT| \geq 2(w-2), $$
$$ |T|=2w+2h \leq 2w+4 \leq 10(w-2) \leq 5|OPT|. $$ 

\qquad {\bf Case 2.1.b:} $h \geq 2$. 
Since each region has unit diameter, the optimal tour may lie
inside $Q$ at distance at most one from its boundary, so
it must touch all four sides of some rectangle $Q'$ having width
$\geq w-2$ and height $\geq h-2$. Hence
$$ |OPT| \geq 2\sqrt{(w-2)^2+(h-2)^2}. $$
Since $8w+8h \geq 40$, we have
$$ |T|=2w+2h \leq 10((w-2)+(h-2)) \leq 10\sqrt{2}\sqrt{(w-2)^2+(h-2)^2}
\leq 5\sqrt{2}|OPT|.  $$ 

\noindent {\bf Case 2.2:} $D \leq 3$. 
The horizontal projection of each region is at least
$1/\sqrt{2}$ and the rectangle $Q$, of width $w\leq 5$,
is partitioned into 8 subrectangles, each of width $\leq 5/8$,
by the vertical segments $e_1,\ldots,e_7$.
Thus, since $1/\sqrt{2}>5/8$, $T$ visits all
regions. (As in Case~1, each region that is not intersected by the
perimeter of $Q$ lies entirely inside $Q$; hence, it is intersected by
one of the seven vertical segments.) A similar calculation yields 
$$ |OPT| \geq 2 diag(Q)=2 \sqrt{w^2+h^2}. $$
$$ |T|=2w+16h \leq 2\sqrt{1^2+8^2} \sqrt{w^2+h^2} \leq 2(8.1)\sqrt{w^2+h^2} 
\leq 8.1|OPT|. $$

\noindent {\bf Case 3}. Partition the optimal tour $OPT$ into blocks $OPT_i$,
with $i \geq 1$. $OPT_1$ starts at an arbitrary point of intersection
of $OPT$ with the leftmost covering line, and ends at the last intersection of $OPT$ 
with the second from the left covering line, before $OPT$ intersects a
different covering line. 
Notice that $OPT$ does not cross
to the left of the leftmost covering line.
In general, the blocks $OPT_i$ are determined by the 
last point of intersection of $OPT$ with a covering line, before $OPT$ 
crosses a different covering line. 
In Figure \ref{cover}, the hypothetical optimum tour is partitioned
into six blocks $AB, BC, CD, DE, EF, FA$. 
For example, $OPT$ crosses the second vertical line twice before
crossing the third vertical line, and $B$ is the last point of
intersection with the second line. 
Consider the bounding box $Q$ of $OPT_i$, 
the smallest perimeter axis-aligned rectangle which includes
$OPT_i$. Write $w$ for its width and $h$ for its height. 
There are two cases to consider. \\

\noindent {\bf Case 3.1:} $OPT_i$ intersects regions stabbed by two
consecutive covering lines only, $l_1,l_2$ say, at distance $w_1$.
This implies that $OPT_i$ lies between these two covering lines, 
so $w=w_1$. Without loss of generality, $OPT_i$ 
touches the lower side of $Q$, at some point $C$, before it touches 
the upper side of $Q$, at some point $D$.  
Since the horizontal projection of each region is at least
$1/\sqrt{2}$, we have $w_1 \geq 1/\sqrt{2}$.
Let $0 \leq a,b \leq h$ specify the starting and ending points 
$A$ and $B$ of $OPT_i$; see Figure~\ref{l2}. 
Namely, $a$ is the vertical distance between $A$, the start point of 
$OPT_i$, and the lower horizontal side of $Q$, and $b$ is the 
vertical distance between $B$, the end point of $OPT_i$, and the upper
horizontal side of $Q$. By considering the reflections of $OPT_i$ with
respect to the two horizontal sides of $Q$, we get that $|OPT_i|$ is 
bounded from below by the length of the polygonal line $A'CDB'$:
$$ |OPT_i| \geq |A'CDB'| \geq \sqrt{(h+a+b)^2+{w_1}^2}. \eqno(5) $$

\begin{figure}[htbp]
\centerline {\input{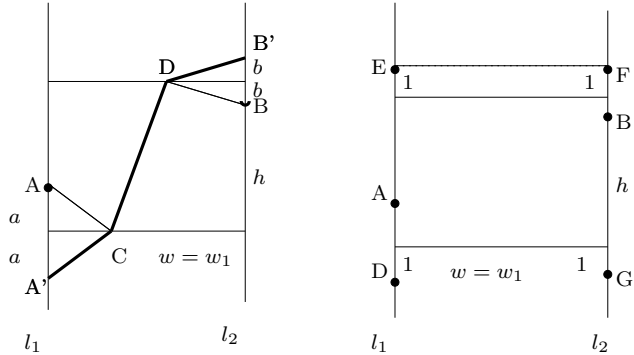}}
\caption{Case 3.1: 2 lines.}
\label{l2}
\end{figure}

We show that there exists a partial tour $T_i$ (a path) of the 
representative points of all the regions $OPT_i$ visits, of length
bounded by $c|OPT_i|$ for some positive constant $c$, where $T_i$
starts at $A$ and ends at $B$. Take $T_i=ADEFGB$, where the points
$D,E,F,G$ are on the lines $l_1,l_2$ at unit distance from the corners
of $Q$; see Figure~\ref{l2}.
$$|T_i| \leq (a+1)+(1+h+1)+(w_1)+(1+h+1)+(1+h-b)
\leq 3(h+a+b)+(w_1+6) $$
$$ \leq 3(h+a+b)+(6\sqrt{2}+1)w_1
\leq \sqrt{3^2+(6\sqrt{2}+1)^2} \sqrt{(h+a+b)^2+{w_1}^2}
< 9.95|OPT_i|. \eqno(6)$$
Put $T_i$ together to get a tour $T$ of the representative points,
having length smaller than $9.95|OPT|$. 
Then, for $\epsilon \leq 0.05$, the $(1+\epsilon)$
approximate tour of representatives, which the algorithm actually
computes, has length at most $10|OPT|$. \\ 

\noindent {\bf Case 3.2:} $OPT_i$ intersects regions stabbed by three
consecutive covering lines only, $l_1$, $l_2$ and $l_3$.
Denote by $w_1$ (resp., $w_2$) the horizontal distance between $l_1$  
and $l_2$ (resp., $l_2$ and $l_3$). 
This implies that $OPT_i$ lies between these three covering lines, 
but it does not touch the rightmost line $l_3$. We can assume that
$OPT_i$  starts at $A$ and ends at $B$, and as in the previous case, 
that it touches the lower side of $Q$, at some point $C$,
before it touches the upper side of $Q$, at some point $D$.
Let $l^{'}$ be the supporting line of the right side of $Q$,
and $w^{'}$ be the horizontal distance between $l_2$ and $l^{'}$.
We have $w_1,w_2 \geq 1/\sqrt{2}$ and $w^{'} \geq \max(0,w_2-1)$.
The case when $OPT_i$ touches the upper side of $Q$ before it touches 
$l^{'}$ is shown is Figure~\ref{l3}; the other case is similar and we
get the same bound on $|OPT_i|$.
Let $a,b,A,B$ be as before. We distinguish two sub-cases:\\

\qquad {\bf Case 3.2.a:} $w_2 \leq 1$. The lower bound on $OPT_i$ we
have used earlier is still valid.
$$ |OPT_i| \geq \sqrt{(h+a+b)^2+{w_1}^2}. $$
We show a partial tour $T_i$ of the 
representative points of all the regions $OPT_i$ visits.
Take $T_i=ADEFBGHIFB$; see Figure~\ref{l3}. Points $E$, $F$ and $I$
(resp., $D$, $G$ and $H$) are on the lines $l_1$, $l_2$ and $l_3$
at unit vertical distance above the upper side (resp., below the lower
side) of $Q$. 
$$ |T_i| \leq (a+1)+(1+h+1)+w_1+(1+h+1)+w_2+(1+h+1)+w_2+(1+b)$$
$$ \leq 3(h+a+b)+(w_1+2w_2+8) \leq 3(h+a+b)+(10\sqrt{2}+1)w_1$$
$$ \leq \sqrt{3^2+(10\sqrt{2}+1)^2} \sqrt{(h+a+b)^2+{w_1}^2}
< 15.45|OPT_i|. \eqno(7) $$
Put $T_i$ together to get a tour $T$ of the representative points,
having length smaller than $15.45|OPT|$. 
Then, for $\epsilon \leq 0.05$, the $(1+\epsilon)$-approximate tour of
representatives has length at most $15.5|OPT|$. \\

\begin{figure}[htbp]
\centerline {\input{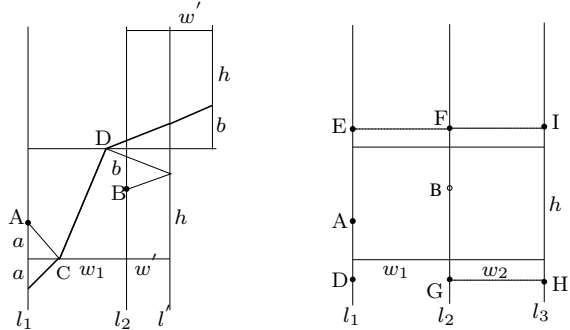}}
\caption{Case 3.2: 3 lines.}
\label{l3}
\end{figure}

\qquad {\bf Case 3.2.b:} $w_2 \geq 1$. We use a different lower 
bound on $OPT_i$.
$$ |OPT_i| \geq \sqrt{(h+a+b)^2+(w_1+2w_2-2)^2},   \eqno(8) $$
which we get by considering the reflections of $OPT_i$ with respect 
to the two horizontal sides of $Q$ and with respect to $l^{'}$;
see Figure~\ref{l3}.
Take $T_i=ADEFBGHIFB$ as in Case~3.2.a.
$$ |T_i| \leq 3(h+a+b)+(w_1+2w_2+8) 
\leq  3(h+a+b)+(10\sqrt{2}+1)(w_1+2w_2-2)$$
$$ \leq \sqrt{3^2+(10\sqrt{2}+1)^2} \sqrt{(h+a+b)^2+(w_1+2w_2-2)^2} 
< 15.45|OPT_i|.  \eqno(9) $$
Put $T_i$ together to get a tour $T$ of the representative points,
having length smaller than $15.45|OPT|$. 
Then, for $\epsilon \leq 0.05$, the $(1+\epsilon)$-approximate tour of
representatives has length at most $15.5|OPT|$. 
The overall approximation ratio of the algorithm, derived from
the Combination Lemma is $15.5+15.5+2=33$.
\end{proof}

\paragraph{Comparison with \cite{ah-aagcs-94}.}
We point out some of the similarities and differences between the
techniques of Arkin and Hassin~\cite{ah-aagcs-94} and our algorithm
and its analysis, which is based on theirs.  First, in
\cite{ah-aagcs-94}, three different algorithms are presented, for
parallel equal-length segments, translates of a convex region, and
translates of an arbitrary connected region.  The second two
algorithms are refinements of the first, and representative points are
chosen differently in each case.  We presented here a single algorithm
that works for all regions of type~(1).  Our consideration of cases 1
and 2 are slightly different from the corresponding cases of
\cite{ah-aagcs-94}, which allows us to handle a larger class of
inputs, namely regions of type (1); e.g. in case 2, \cite{ah-aagcs-94}
distinguishes between the subcases $D \leq 1$ and $D \geq 1$, while we
distinguish between $D \leq 3$ and $D \geq 3$, and the cases are
treated slightly differently.  Finally, the analysis of the algorithms
is based on similar ideas, e.g. to divide the optimal tours into
blocks; our analysis differs in being able to address a single
algorithm and a larger class of inputs.

\subsubsection*{Some Special Cases} \qquad

We note that our calculations of the approximation ratio for
connected regions of a same diameter, give improved bounds
for three cases addressed in \cite{ah-aagcs-94}:
\begin{enumerate}
\item parallel equal segments, 
from $3\sqrt{2}+1$, to $3\sqrt{2}$,
\item translates of a convex region, 
from $\sqrt{3^2+7^2}+1$, to $\sqrt{3^2+7^2}$,
\item translates of a connected region, 
from $\sqrt{3^2+{11}^2}+1$, to $\sqrt{3^2+{11}^2}$.
\end{enumerate}

The reason why these cases can be improved, and the new approximation
ratios have similar expressions with the old ones, is that our
algorithm {\bf A} is similar to the algorithms in
\cite{ah-aagcs-94} for those cases. 

We exemplify here the case of parallel equal segments
and omit details for the rest.
The algorithm computes a greedy cover of the segments, assumed to be of unit
length, using vertical lines. Then it proceeds according to the cardinality
of the cover. Cases 1 and 2 are treated in \cite{ah-aagcs-94}, and the
ratio is $\sqrt{2}$; this is not the bottleneck case.
In Case 1 (one covering line), an optimal tour is easy to obtain.
In Case 2 (two covering lines), a smallest aligned rectangle which 
touches all segments is the output tour.
In Case 3 (three or more covering lines) the algorithm computes an
almost optimal tour of the representative points, as algorithm {\bf A}
does. Its analysis is divided into two sub-cases, as in the proof of 
Theorem ~\protect{\ref{thm:main-approx}}. 
In the first sub-case ($OPT_i$ intersects segments covered by two
consecutive covering lines only), the lower bound in Equation~(5) 
on $|OPT_i|$ is valid. The upper bound in Equation~(6) on $|T_i|$ is 
adjusted by dropping the constant term equal to $+6$. Then
$$ |T_i| \leq 3(h+a+b)+w_1 \leq \sqrt{3^2+1^2}|OPT_i| = \sqrt{10}|OPT_i|. $$
In the second sub-case ($OPT_i$ intersects segments covered by three
consecutive covering lines only), the lower bound in Equation~(8)  
on $|OPT_i|$ is valid. The upper bound  in Equation~(9) on $|T_i|$ is 
adjusted by dropping the constant term equal to $+8$. We also have that 
$w_1, w_2 >1$. Then
$$ |T_i| < 3(h+a+b)+w_1+2w_2 \leq 3(h+a+b)+3(w_1+2w_2-2) \leq
3\sqrt{2}|OPT_i|. $$
The overall approximation ratio obtained for parallel equal segments
is $3\sqrt{2}$.

\section {Lines}
\label{sec:lines}

We consider now the case in which the $n$ regions defining the TSPN
instance are infinite straight lines in the plane.  It is interesting
that this case allows for an exact solution in polynomial time:

\begin{proposition} \label{P3}
Given a set $L$ of $n$ infinite straight lines in the plane, a shortest 
tour that visits $L$ can be computed in polynomial time.
\end{proposition}

\begin{proof} 
We convert the problem to an instance of the {\em watchman route
problem} in a simple polygon.
A {\em watchman route} in a polygon is a tour inside the polygon, such
that every point in the polygon is visible from some point along the
tour. The watchman route problem asks for a watchman route of minimum
length \cite{AGS00}. The problem is known to have an
$O(n^5)$ algorithm; see \cite{t-01}, as well as \cite{cjn-99}. 
Let $B$ be a rectangle that contains all of the vertices of the
arrangement of $L$.  At any one of the two points of intersection
between a line $l_i\in L$ and the boundary of $B$, we extend a very
narrow ``spike'' outward from that point, along $l_i$, for some fixed
distance.  Let $P$ be the simple polygon having $3n+4$ vertices that
is the union of $B$ and these $n$ spikes, as illustrated in
Figure~\ref{watch}.  We make the observation that a tour $T$ visits
all of the lines in $L$ if and only if it sees all of the polygon $P$,
as required by the watchman route problem.  Consequently, we can solve
the TSPN on the set of lines $L$ by solving the watchman route problem
on polygon~$P$.  
\end{proof}

\begin{figure}[htbp]
\centerline{\epsfysize=2.2in \epsffile{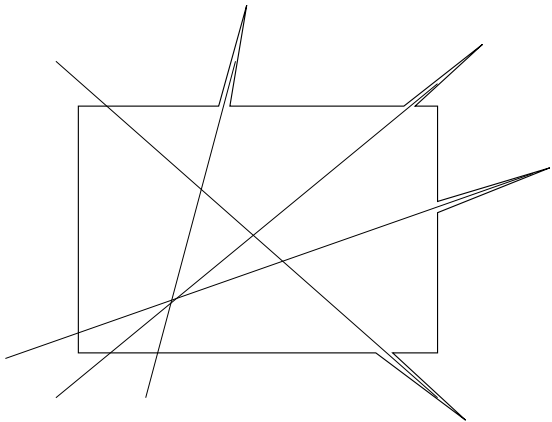}}
\caption{Proof of Proposition \ref{P3}.}
\label{watch}
\end{figure}
 
Given the high running time of the watchman route algorithms, it is of
interest to consider more efficient algorithms that may approximate
the optimal solution.  To this end, we now present a {\em linear-time}
constant-factor approximation algorithm.  

Let $L=\{l_1,\ldots,l_n\}$ be the input set of $n$ lines. A {\em
minimum touching circle (disk)} is a circle of minimum radius which
intersects all of the lines in $L$.  The algorithm computes
and outputs $C_L$, a minimum touching circle for $L$.  We will show
that this provides a tour of length at most $\frac{\pi}{2}|OPT|$,
where, as usual, $OPT$ denotes an optimal tour.  First we argue about
the approximation ratio, and leave for later the presentation of the
algorithm. To start, we assume for simplicity that no two lines are parallel,
though this assumption will be later removed.

\begin{observation} \label {O1}
The optimal tour of a set of lines is a (possibly degenerate) convex
polygon $P$. 
\end{observation}

\begin{proof}
It is easy to see that an optimal tour must be polygonal,
consisting of a finite union of straight line segments.  If an optimal
tour were a non-convex polygon $P$, we obtain a contradiction to its
optimality since the boundary of the convex hull of $P$, which is
shorter than $P$, also visits all of the lines $L$, since $P$ does.
\end{proof}

\begin{observation} \label {O2}
$C_L$ is determined by $3$ lines in $L$, i.e. it is the inscribed circle 
in the triangle $\Delta$ formed by these $3$ lines. 
\end{observation}

We distinguish two cases:\\

\noindent {\bf Case 1.} $\Delta$ is an acute triangle. It is well known
that for an acute triangle, the minimum perimeter inscribed triangle
(having a vertex on each side of the triangle) is its {\em pedal}
triangle, whose vertices are the feet of the altitudes of the given
triangle (see e.g. \cite{RT57}).
So in this case, the optimal tour $OPT_{\Delta}$, which visits the $3$ 
lines of $\Delta$ is its pedal triangle; we denote by $y$ its perimeter. 
Clearly $y$ is a lower bound on $|OPT|$: $|OPT| \geq |OPT_{\Delta}|=y $.
Denote by $s$ the semi-perimeter of $\Delta$, by $R$ the radius of 
its circumscribed circle, and by $r$ the radius of its inscribed circle.

\begin{fact} \label{F2}
For an acute triangle $\Delta$, $s>2R$.
\end{fact}
\begin{proof}
If $A,B,C$ are the angles of $\Delta$, this is equivalent to
$$ R(\sin{A}+\sin{B}+\sin{C}) > 2R. $$
After simplification with $R$, this is a well known inequality in
the geometry of an acute triangle (\cite{bdjmv-gi-69}, page 18).
\end{proof}

\begin{fact} \label{F3}
For an acute triangle $\Delta$, $y=\frac{2rs}{R}$.
\end{fact}
\noindent A proof of this equality can be found in \cite{bdjmv-gi-69}, 
page 86. 

\begin{claim} \label{C1} 
For an acute triangle $\Delta$, $r<\frac{y}{4}$.
\end{claim}
\begin{proof} 
Putting the above together, we get 
$ r=\frac{Ry}{2s}<\frac{Ry}{4R}=\frac{y}{4}$. 
\end{proof}

\noindent The length of the output tour is bounded as follows

$$ |C_L|=2\pi r < 2\pi\frac{y}{4}= \frac{\pi}{2}y \leq \frac{\pi}{2}|OPT|. $$ 

\noindent {\bf Case 2.} $\Delta$ is an obtuse triangle. In this case,
$|OPT_{\Delta}|=2h$, where $h$ is the length of the altitude corresponding
to the obtuse angle, say $A$. 
Clearly $2h$ is a lower bound on $|OPT|$: $|OPT| \geq |OPT_{\Delta}|=2h $.

\begin{fact} \label{F4}
For any triangle $\Delta$, $h>2r$.
\end{fact}
\begin{proof}
Let $S$ denote the area of the triangle with side lengths
$a$, $b$, and $c$.  Then, we know that $S=ah/2$, by our definition of
the altitude $h$.  We also know from elementary geometry that
$S=(a+b+c)r/2$, recalling that $r$ is the radius of the inscribed circle
of the triangle.  Thus, $ah= (a+b+c)r$, from which we obtain
$$\frac{h}{r}=\frac{a+b+c}{a}>\frac{a+a}{a}=2,$$
using the triangle inequality.
\end{proof}

\noindent Using this inequality, we have $ |OPT_{\Delta}|=2h > 4r $. 
The length of the output tour is bounded as follows
$$ |C_L|=2\pi r =\frac{\pi}{2}(4r) < \frac{\pi}{2}|OPT_{\Delta}|
\leq \frac{\pi}{2}|OPT|. $$ 
Thus, in both cases, the approximation ratio is 
$\frac{\pi}{2} \leq 1.58$.

When $C_L$ is determined by 3 lines in $L$ of which two are parallel
at distance $h$, we say that they form a generalized triangle
$\Delta$. 
Clearly $2h$ is a lower bound on $|OPT|$: $|OPT| \geq |OPT_{\Delta}|=2h $.
We also have $r=h/2$, where $r$ is the radius of $C_L$, thus
$$ |C_L|=2\pi r = \pi h \leq \frac{\pi}{2}|OPT|. $$ 

We now describe the algorithm for computing $C_L$.
The distance $d(p,l)$ from a point $p$ of coordinates $(x_0,y_0)$ 
to a line $l$ of equation $ax+by+c=0$ is 
$$ d(p,l)=\frac{|ax_0+by_0+c|}{\sqrt{a^2+b^2}}.$$
Let the lines in $L$ have the equations
$$ (l_i:) \ \ \ \ \ \ a_ix+b_iy+c_i=0, \ i=1,\ldots,n. $$
Finding a minimum touching circle amounts to finding 
the center coordinates $(x,y)$, and a minimum radius $z$, s.t. 
$$ \frac{|a_ix+b_iy+c_i|}{\sqrt{a_i^2+b_i^2}} \leq z, \ \ i=1,\ldots,n. $$
This is equivalent to solving the following $3$-dimensional linear
program 
\smallskip
$$ \min z {\rm \ \ subject \ to \ \ } \left \{
\frac{a_ix+b_iy+c_i}{\sqrt{a_i^2+b_i^2}} \leq  z, \ \ 
\frac{-a_ix-b_iy-c_i}{\sqrt{a_i^2+b_i^2}} \leq  z, \ \ 
i=1,\ldots,n \right \}, $$
\noindent which takes $O(n)$ time~\cite{m-lpltw-84}.
Consequently, we have proved
\begin{theorem} \label{P4}
Given a set $L$ of $n$ (infinite straight) lines in the plane, 
a $\frac{\pi}{2}$-approximate tour that visits $L$ can be computed in  
$O(n)$ time.
\end{theorem}

\section{Conclusion}
\label{sec:conclusion}

Several open problems remain, including
\begin{description}
\item[(1)] Is there a constant-factor approximation algorithm for arbitrary
connected regions in the plane?  What if the regions are disconnected? 
(giving us a geometric version of a ``one-of-a-set TSP'')
\item[(2)] What approximation bounds can be obtained in higher dimensions?
Our packing arguments for disjoint disks lift to higher dimensions, but our
other methods do not readily generalize.

A particularly intriguing special case is the generalization of the case of
infinite straight lines: What can be said in 3-space for the TSPN on a set 
of lines or of planes?

\item[(3)] Is there a PTAS for general pairwise-disjoint regions in the plane?
\end{description}

\subsection*{Acknowledgements}

We thank Estie Arkin and Michael Bender for several useful discussions
on the TSPN problem. 
We thank the anonymous referees for their detailed 
comments and suggestions, which greatly improved the paper.

\end{document}